\documentclass[9pt]{article}
\usepackage{spconf,amsmath,graphicx}
\usepackage{epstopdf}
\usepackage{slashed}
\usepackage{bm}
\usepackage{mdwmath}
\usepackage{amsthm}
\usepackage{url}
\usepackage{amssymb}
\usepackage{enumerate}
\usepackage{algorithm} 
\usepackage{algorithmic} 
\usepackage{subfigure}
\usepackage{appendix}
\usepackage{cite}
\usepackage{makecell,rotating,multirow,diagbox}
\usepackage{color}

\newtheorem{proposition}{\textbf{Proposition}}

\graphicspath{{figures/}}

\def\y{{\mathbf y}}
\def\0{{\mathbf 0}}

\def\p{{\mathbf p}}
\def\f{{\mathbf f}}
\def\g{{\mathbf g}}
\def\r{{\mathbf r}}
\def\n{{\mathbf n}}

\def\I{{\mathbf I}}

\def\A{{\mathbf A}}

\def\Q{{\mathbf Q}}
\def\U{{\mathbf U}}
\def\C{{\mathbf C}}

\def\Tr{{\textrm {Tr}}}

\def\bphi{{\pmb {\phi}}}

\def\cS{{\mathcal S}}

\def\cO{{\mathcal O}}

\def\bbR{{\mathbb R}}

\def\bphi{{\pmb {\phi}}}

\def\cS{{\mathcal S}}

\def\cO{{\mathcal O}}

\def\bphi{\pmb {\Phi}}
\def\bsphi{\pmb {\phi}}
\def\bdelta{\pmb {\delta}}

\def\bsigma{\pmb{\Sigma}}
\def\U{\mathbf{U}}
\def\V{\mathbf{V}}
\def\Tr{{\textrm {Tr}}}
\def\argmax{{\textrm {argmax}}}

\title{Fast sensor placement by enlarging principle submatrix for large-scale linear inverse problems}
\name{{Fen Wang$^{\dag}$, Gene Cheung$^{\ddag}$, Taihao Li$^{\dag,*}$, Ying Du$^{\dag}$, Yu-Ping Ruan$^\dag$}
\thanks{$^*$Corresponding author}
}
\address{$^{\dag}$Department of Artificial Intelligence,  Zhejiang Lab, Hangzhou, China\\
$^{\ddag}$Department of Electrical Engineering \& Computer Science, York University, Canada
}

\begin{document}
%
\maketitle

\begin{abstract}
Sensor placement for linear inverse problems is the selection of locations to assign sensors so that the entire physical signal can be well recovered from partial observations.  
In this paper, we propose a fast sampling algorithm to place sensors. 
Specifically, assuming that the field signal $\f$ is represented by a linear model $\f=\bphi{\g}$, it can be estimated from partial noisy samples via an unbiased least-squares (LS) method, whose expected mean square error (MSE) depends on chosen samples.
First, we formulate an approximate MSE problem, and then prove it is equivalent to a problem {related to a principle submatrix of $\bphi\bphi^\top$ indexed by sample set. }
To solve the formulated problem, we devise a fast greedy algorithm with simple matrix-vector multiplications, leveraging a matrix inverse formula.
To further reduce complexity, we reuse results in the previous greedy step for warm start, so that candidates can be evaluated via lightweight vector-vector multiplications.
Extensive experiments show that our proposed sensor placement method achieved the lowest sensor sampling time and the best performance compared to state-of-the-art schemes.



\end{abstract}

\section{Introduction}
\label{sec:introd}
Sensor placement is the problem choosing a budget of locations for sensors to monitor a physical field, such as temperature, humidity and transportation congestion \cite{weak-submodularity2018,Humidity,traffic}.
Due to limited sensing resources and expensive acquisition cost, sampling techniques are widely studied to evaluate possible locations to place sensors for collecting signals, so that the unobserved signal can be recovered in high accuracy \cite{ortiz2019sparse,anis16,fen2020,Chamon2018}.

Assuming the physical signal is represented by a linear model $\f=\bphi\g$ \cite{linear}, it can be  recovered from partial noisy samples by solving a \textit{linear inverse problem} using least square (LS) method, whose \textit{mean square error} (MSE) is a function of measurement matrix $\bphi$ and sampling set \cite{Statistical-signal}.  
Recently, authors in \cite{joshi2008sensor} formulated the sensor placement problem as a nonconvex
optimization problem and then relaxed it into a convex one, so that it can be solved by interior point methods with polynomial complexity. 
However, this relaxation scheme performed poorly when the sensor budget was very small. 
To promote better MSE performance,
greedy algorithms were devised to design sensor locations one-by-one via solving different local optimization problems
\cite{linear-inverse,jiang2016sensor,jiang2019group}. 
\cite{linear-inverse} presented one  near-optimal sensor placement
algorithm for the frame potential-based cost function, called FrameSense, which preserved the sub-modularity property with theoretically bounded performance. 
\cite{jiang2016sensor} used the worst-case of MSE function---E-optimality criterion \cite{experiments}---as objective to deploy samples, and developed
corresponding efficient algorithms. 
However, in each greedy search, it required computing eigenspace of selected sensors, which is
expensive when the dimension of eigenspace is large.
Recently, \cite{jiang2019group} proposed a fast MSE pursuit algorithm to greedily solve an approximate MSE criterion, but it still suffered from multiple matrix inverse computations, which is not practical for large-scale linear inverse problems. 

{Towards fast sensor placement, in this paper we propose a greedy approach with simple vector-vector multiplications. 
Specifically, first we formulate an approximate MSE problem using a small shift, and then prove it is equivalent to a problem related to a principle submatrix of $\bphi\bphi^\top$ indexed by samples.
To solve the formulated problem, we devise a fast greedy algorithm without matrix inverse computation based on a matrix inverse lemma \cite{blockwise}.
Using computed results in the previous greedy step for warm start, we design an accelerated strategy to evaluate each candidate via simple vector-vector multiplications.
Extensive experiments show that among popular sensor placement methods, our proposed scheme achieved the lowest sensor sampling time with the best sampling performance.}
\section{Preliminaries}
In sensor placement field,  the high-dimensional physical signal $\f\in\bbR^{N}$ is conventionally modelled by a linear equation \cite{Sensor-placement}: 
\begin{equation}
    \f=\bphi\g,
\end{equation}
where $\g\in\bbR^K$ is the parameter vector where $K\ll N$. 

Instead of placing $N$ sensors, given limited measuring resources, we select $M$ samples out of $N$ total choices, where $M \ll N$. 
We first define a sampling matrix $\C\in\{0,1\}^{M\times N}$ associated with sampling set $\cS\subset \{1,\dots,N\}$ as follows \cite{SCsampling}:
\begin{equation}\label{eq:sampling-matrix}
\begin{split}
  c_{ij}= \left\{ \begin{array}{ll}
1,&\mbox{if}\; j=\cS\{i\}; \\
0,&\mbox{otherwise}.
\end{array}\right.
\end{split}
\end{equation}
Thus, the observed noisy signal can be expressed as
\begin{equation}\label{eq:sampled-signal}
  \y=\C\f+\n=\C\bphi\g+\n\in\bbR^M,
\end{equation}
where $\n$ is an i.i.d noise signal with covariance matrix $\sigma^2\I$.

Using observed partial signal $\y$, we can estimate the original signal $\f$ by solving a {linear inverse problem}. 
Specifically, assuming the coefficient matrix $\C\bphi$ is full column rank{\footnote{One necessary but not sufficient condition of $\textrm{rank}(\C\bphi)=K$ is $M\ge K$ \cite{Anis2018TIT}.}}, the parameter vector can be estimated as \cite{Statistical-signal} 
\begin{equation}
     \hat{\g}=(\C\bphi)^\dag\y,
\end{equation}
where $\dag$ denotes the pseudo inverse computation, and $\hat{\g}$ is an unbiased estimator of $\g$ with minimum variance. 
Given $\hat{\g}$, the estimated target signal is $\hat{\f}=\bphi\hat{\g}$. 

Since $(\C\bphi)^\dag=\left[(\C\bphi)^\top(\C\bphi)\right]^{-1}(\C\bphi)^\top$, the expected MSE of this least square (LS) solution is \cite{wang2019low}
\begin{equation}\label{eq:mse}
  \textrm{MSE}(\hat{\g})=\mathbb{E}\left(\|\hat{\g}-\g\|^2_2\right)
  =\sigma^2\Tr\left[(\C\bphi)^\top(\C\bphi)\right]^{-1} .
\end{equation}

Hence, a MSE-based sampling problem for placing sensors is to select $M$ samples to minimize the expected MSE\footnote{A-optimality formulation minimizes the expected MSE directly, while other cost functions, like E-optimality or D-optimality criterion, minimize proxies of MSE \cite{experiments}.}.
This is also called \textit{A-optimality} in experimental design \cite{experiments}:
\begin{equation}\label{eq:mse-sampling}
  \min_{\C}~~\Tr\left[(\C\bphi)^\top(\C\bphi)\right]^{-1}=\sum_{k=1}^{K}\frac{1}{\lambda_k}
\end{equation}
where $\lambda_K\ge \dots \ge \lambda_2\ge \lambda_1$ are the eigenvalues of matrix $(\C\bphi)^\top(\C\bphi)$ \cite{wang2019low}.
{Since $\C\bphi$ has full-column rank and $(\C\bphi)^\top(\C\bphi)$ is positive semi-definite (PSD) by definition, $(\C\bphi)^\top(\C\bphi)$ has the property that $\lambda_k> 0,\forall k$.}

\section{Fast Sampling Strategy to minimize an approximate MSE problem}

{In this section, we at first propose an augmented A-optimality criterion as our sampling objective.
Then, we propose to mitigate the large matrix inverse in each greedy step based on matrix inverse formula.
For fast sampling, we propose one strategy to reduce the computation burden based on warm starts.
At last, we will analyze the complexity of our proposed sampling algorithm.

\subsection{{Modified A-optimality Criterion}}
\label{subsec:modifiedA}


{First, we propose a modified A-optimality sampling criterion that closely approximates the original problem \eqref{eq:mse-sampling} by adding a small constant shift:}
\begin{equation}\label{eq:mse-shift}
  \min_{\C}~~\Tr\left[(\C\bphi)^\top(\C\bphi)+\mu\I\right]^{-1}=\sum_{k=1}^{K}\frac{1}{\lambda_k+\mu}
\end{equation}
where $\mu>0$ is a small shift parameter.
{This shifted sampling objective was also adopted in sensor placement \cite{jiang2019group} and graph sampling \cite{wang2019low}.
We will present a faster algorithm minimizing this objective than ones in \cite{jiang2019group} and \cite{wang2019low} in this section.}


{It can be proven that the modified optimization \eqref{eq:mse-shift} has the same optimal solution(s) if $(\bphi\bphi^\top+\mu\I)_{\cS}$ replaces $( \C \bphi )^{\top} (\C \bphi)+\mu\I$ in \eqref{eq:mse-shift}.
We formally state this in the following proposition:}

\begin{proposition}
\label{prop:submatrix}
Denote the eigenvalues of matrix $(\C\bphi)^\top(\C\bphi)$ by  $\lambda_1\le \lambda_2\le \dots \le \lambda_K$, and assume $\textrm{rank}(\C\bphi)=K$.
When $|\cS|=M\ge K$, an optimal sampling set for problem  \eqref{eq:mse-shift} is also optimal to the following optimization:
  \begin{equation}\label{eq:sampling-submatrix}
\min_{\cS}~~\Tr\left[(\bphi\bphi^\top+\mu\I)_{\cS}\right]^{-1}
  \end{equation}
where $\A_{\cS}$ is a principle submatrix of $\A$ with row and column indexed by set $\cS$ and the relationship between $\C$ and $\cS$ is defined in equation \eqref{eq:sampling-matrix}.
\end{proposition}

\begin{proof}
{{
When $M=K$, matrix $\C\bphi$ is a square matrix, so the eigenvalues of $(\C\bphi)^{\top} (\C\bphi)$ and $(\C\bphi)(\C\bphi)^{\top}=(\bphi\bphi^\top)_{\cS}$ are the same. Therefore, problem \eqref{eq:mse-shift} and \eqref{eq:sampling-submatrix} have equivalent  optimal solution;
When $M>K$, assume first that the singular value decomposition (SVD) of matrix $\C\bphi\in\bbR^{M\times K}$ is $\C\bphi=\U\bsigma\V^\top$ where $\U\in\bbR^{M\times M}$ and $\V\in\bbR^{K\times K}$ are orthogonal left and right eigenvectors respectively. The singular value  matrix $\bsigma\in\bbR^{M\times K}$ has the form:
  \begin{equation}\label{eq:sigma}
    \bsigma=\left[\begin{array}{ccc}
                    \sigma_1 & ~ & ~ \\
                    ~ & \ddots & ~ \\
                    ~ & ~ & \sigma_K \\
                    ~ & \0 & ~
                  \end{array}\right]
  \end{equation}
where $\sigma_k\neq 0, \forall k$, since $\textrm{rank}(\C\bphi)=K$ and $M>K$.

Then, the eigenvalues of $(\C\bphi)^{\top} (\C\bphi)=\V\bsigma^\top\bsigma\V^\top$ are $\{\sigma^2_1,\cdots,\sigma^2_K\}$, where $\lambda_k=\sigma^2_k>0,\forall k$.
And, the eigenvalues of matrix $(\bphi\bphi^\top)_{\cS}=(\C\bphi)^{\top} (\C\bphi)=\U\bsigma\bsigma^\top\U^\top$ are $\{\underbrace{0,\dots,0}_{M-K},\lambda_1,\dots,\lambda_K\}$}}, which indicates the eigenvalues of $(\bphi\bphi^\top)_{\cS}+\mu\I$ are
  \begin{equation}\label{eq:submatrix-eigenvalue}
    \{\underbrace{\mu,\dots,\mu}_{M-K},\lambda_1+\mu,\dots,\lambda_K+\mu\}
  \end{equation}

Therefore,
\begin{equation}
  \Tr\left[(\bphi\bphi^\top)_{\cS}+\mu\I\right]^{-1}=\frac{M-K}{\mu}+\sum_{k=1}^{K}\frac{1}{\lambda_k+\mu}
\end{equation}
where $\lambda_k+\mu > 0$ since $\mu$ is strictly positive and $\lambda_k>0$.

Combined with equation \eqref{eq:mse-shift}, we see that
\begin{equation}
  \Tr\left[(\bphi\bphi^\top)_{\cS}+\mu\I\right]^{-1}=\frac{M-K}{\mu}+\Tr\left[(\C\bphi)^\top(\C\bphi)+\mu\I\right]^{-1}\nonumber
\end{equation}
whose left part can be rewritten as equation \eqref{eq:sampling-submatrix}.

Notice that i) $M$, $K$ and $\mu$ are fixed constants not affected by optimization, and ii) optimization variables $\C$ in \eqref{eq:mse-shift} or $\cS$ in \eqref{eq:sampling-submatrix} have the same degrees of freedom.
We can thus conclude that optimal solutions to \eqref{eq:mse-shift} and to \eqref{eq:sampling-submatrix} are the same.
\end{proof}
{\textbf{Remark:} This formulation was also derived using Neumann series theorem \cite{neumannseries} for sampling of bandlimited graph signals, where $\bphi$ is formed using the first $K$ orthogonal eigenvectors of a combinatorial graph Laplacian matrix in \cite{wang2019low}.
In contrast, in this paper $\bphi$ is a general measurement matrix for modelling physical filed, which in general does not have the orthogonal property.
Thus, we derive the more general result differently here.  }

\subsection{Greedy Reformulation without Matrix Inversion}
\label{subsec:greedy}
Given that the sampling problem \eqref{eq:sampling-submatrix} is combinatorial in nature, we employ a greedy approach to optimally choose one sample at a time.
Specifically, assuming that we have obtained set $\cS_t$ after $t$ iterations, to decide the $(t+1)$-th sample, we solve the following local optimization problem:
\begin{equation}
\label{eq:submatrix-greedy}
\min_{i\in\cS^c_t}~~\underbrace{\Tr\left[(\bphi\bphi^\top+\mu\I)_{\cS_t\cup\{
i\}}\right]^{-1}}_{f(\cS_t\cup\{
i\})}
\end{equation}
where $\cS_0=\emptyset$.

If we compute the objective directly to evaluate each candidate solution $i\in\cS^c_t$, we have to perform matrix inversion, with complexity at most $\cO (M^3)$.
To mitigate large matrix inversion, we introduce the next greedy strategy based on one matrix inverse formula \cite{blockwise}.

For notation simplicity, we first define $\Q=\bphi\bphi^\top+\mu\I$.
Since matrix $\Q$ is symmetric, with appropriate permutation, its sub-matrix $\Q_{\mathcal{S}_t\cup\{i\}}$ can be expressed as
\begin{align}
\Q_{\mathcal{S}_t\cup\{i\}} = \left[ \begin{array}{cc}
 \mathbf{Q}_{\mathcal{S}_t}&{\mathbf{Q}}_{\mathcal{S}_t,\{i\}} \\
{\mathbf{Q}}_{\{i\},\mathcal{S}_t}&q_{ii} \\
 \end{array} \right]= \left[ \begin{array}{cc}
 \mathbf{Q}_{\mathcal{S}_t}&\mathbf{p}_{t,i} \\
~\mathbf{p}^{\top}_{t,i}&q_{ii} \\
 \end{array} \right],
\end{align}
where $\p_{t,i}={\mathbf{Q}}_{\mathcal{S}_t,\{i\}}\in\bbR^{t}$.

Matrix inverse $\Q^{-1}_{\cS_t\cup\{
i\}}$ in equation \eqref{eq:submatrix-greedy} can be obtained using $\Q^{-1}_{\cS_t}$ via the \textit{matrix inversion formula} \cite{blockwise}:
\begin{small}
\begin{align}\label{inversion lemma1}
\mathbf{Q}^{-1}_{\mathcal{S}_t\cup\{i\}} = \left[ \begin{array}{cc}
\mathbf{Q}^{-1}_{\mathcal{S}_t}+
h^{-1}_i\mathbf{Q}^{-1}_{\mathcal{S}_t}\mathbf{p}_{t,i}\mathbf{p}^{\top}_{t,i}\mathbf{Q}^{-1}_{\mathcal{S}_t}
&-h^{-1}_i\mathbf{Q}^{-1}_{\mathcal{S}_t}\mathbf{p}_{t,i} \\
-h^{-1}_i\mathbf{p}^{\top}_{t,i}\mathbf{Q}^{-1}_{\mathcal{S}_t}&h^{-1}_i \\
 \end{array} \right]
\end{align}
\end{small}
where $h_i=q_{ii}-\mathbf{p}^{\top}_{t,i}\mathbf{Q}^{-1}_{\mathcal{S}_t}\mathbf{p}_{t,i}$ is a scalar.

Therefore,
\begin{equation}\label{eq:cost}
\begin{split}
 f(\cS_t\cup\{i\})&=\Tr\left(\Q^{-1}_{\cS_t}\right)
+h^{-1}_i\Tr\left(\mathbf{Q}^{-1}_{\mathcal{S}_t}\p_{t,i}\mathbf{p}^{\top}_{t,i}\mathbf{Q}^{-1}_{\mathcal{S}_t}\right)+h^{-1}_i\\
&=f(\cS_t)+h^{-1}_i\|\mathbf{Q}^{-1}_{\mathcal{S}_t}\mathbf{p}_{t,i}\|^2_2+h^{-1}_i\nonumber
\end{split}
\end{equation}

Because $f(\cS_t)$ {is a constant not affected by the selection of candidate $i$}, during the $(t+1)$-th greedy step given input $\cS_t$, the sampling problem \eqref{eq:submatrix-greedy} can be simplified as:
\begin{equation}
\label{eq:greedy-mse-shift}
\begin{split}
  \min_{i\in\cS^c_t}~~~&h^{-1}_i\|\mathbf{Q}^{-1}_{\mathcal{S}_t}\mathbf{p}_{t,i}\|^2_2+h^{-1}_i\\
\text{s.t.}~~~&h_i=q_{ii}-\p^{\top}_{t,i}\Q^{-1}_{\mathcal{S}_t}\p_{t,i};~\p_{t,i}=\Q_{\mathcal{S}_t,\{i\}}
  \end{split}
\end{equation}
where $\Q^{-1}_{\cS_t}$ is {already computed during the previous iteration}.

Compared to solving problem \eqref{eq:submatrix-greedy}, optimizing problem \eqref{eq:greedy-mse-shift} needs to compute {matrix-vector product} $\Q^{-1}_{\mathcal{S}_t}\p_{t,i}$ with complexity $\cO(t^2)$, given known $\Q^{-1}_{\cS_t}$.

To optimize problem  \eqref{eq:greedy-mse-shift}, we can compute  $\Q=\bphi\bphi^\top+\mu\I$ once with complexity $\cO(KN^2)$ and then query its partial entries $q_{ii}$ and $\p_{t,i}$ for greedy evaluation.
Next, we will propose to compute the involved elements on the fly inside greedy search without first computing $\Q$ and further reduce the evalution complexity.

\subsection{{Evaluation Complexity Reduction}}
\label{subsec:Q}

We first write input matrix $\bphi=[\bsphi_1,\bsphi_2,\dots,\bsphi_N]^\top$, where $\bsphi^\top_i$ is the $i$-th row in matrix $\bphi$.
Then we can compute $q_{ii}$ as
\begin{equation}
\label{eq:qii}
  q_{ii}=\bdelta^\top_i\Q\bdelta_i=\bdelta^\top_i\bphi\bphi^\top\bdelta_i+\mu=\|\bsphi_i\|^2_2+\mu
\end{equation}

The value of $\p_{t,i}$ can be obtained via
\begin{equation}
\label{pti}
  \p_{t,i}=\Q_{\mathcal{S}_t,\{i\}}=\C_{(\cS_t)}(\bphi\bphi^\top+\mu\I)\bdelta_i=\C_{(\cS_t)}\bphi\bsphi_i+\mu\C_{(\cS_t)}\bdelta_i\nonumber
\end{equation}
where $\bdelta_i$ is the $i$-th column of identity matrix $\I$, and $\C_{(\cS_t)}$ is the sampling matrix corresponding to set $\cS_t$, defined in equation  \eqref{eq:sampling-matrix}.

Based on the definition of sampling matrix,
\begin{equation}\label{eq:prod}
  \C_{(\cS_t)}\bphi\bsphi_i=\bphi_{\cS_t,:}\bsphi_i=[\bsphi_{\cS_t\{1\}},\dots,\bsphi_{\cS_t\{t\}}]^\top\bsphi_i
\end{equation}
and
\begin{equation}\label{eq:c0}
  \mu\C_{(\cS_t)}\bdelta_i=\0
\end{equation}
since $i\in\cS^c_t$.


\begin{algorithm}[tp]
\caption{Fast MSE-based sampling (FMBS)}
\label{algo:FMBS}
\textbf{Input:} $\bphi=[\bsphi_1,\bsphi_2,\dots,\bsphi_N]^\top$, sample size $M$ and $\mu$\\
\textbf{Initialization:} $\cS=\emptyset$
\begin{algorithmic}[1]
\STATE Compute $q_{ii}=\|\bsphi_i\|^2_2+\mu,\forall i$
\STATE Select the first node by $i^* = \argmax_iq_{ii}$
\STATE Update $\cS \leftarrow \cS \cup \left\{ i^* \right\}$
\STATE\textbf{While} $\left| \mathcal{S} \right| < M$
 \STATE $\forall i \in \mathcal{S}^{c}$, compute \\
\textbf{If} $|\cS|=1$\\
 \qquad $\p_{i}=\bsphi^\top_{i^*}\bsphi_i$ and  $\r_i=\frac{1}{q_{i^*i^*}}\p_i$\\
 \textbf{else}\\
 \qquad $\alpha=\p^\top_i\r_{i^*}/h_{i^*}$ and $\beta=\bsphi^\top_{i^*}\bsphi_i/h_{i^*}$ \\
 \qquad $\r_i=\left[\begin{array}{c}
                     \r_i+\alpha\r_{i^*}-\beta\r_{i^*} \\
                     -\alpha+\beta
                     \end{array}\right]$\\
 \qquad $\p_i=[\p^\top_{i} ~~\beta h_{i^*}]^\top$\\
  \textbf{end If}\\
 $h_i=q_{ii}-\p^\top_i\r_i$
\STATE Select $i^* = \mathop {\arg \min }\limits_{i \in \mathcal{S}^{c}} h^{-1}_i\|\r_i\|^2_2+h^{-1}_i $\\
\STATE Update $\cS \leftarrow \cS \cup \left\{ i^* \right\}$
\STATE \textbf{end While}
\STATE Return $\mathcal{S}$
\end{algorithmic}
\end{algorithm}

Using \eqref{eq:qii} to \eqref{eq:c0}, we can simplify optimization \eqref{eq:greedy-mse-shift} to the following
\begin{equation}
\label{eq:finalequ}
\begin{split}
 \min_{i\in\cS^c_t}\quad&h^{-1}_i\|\mathbf{r}_{t,i}\|^2_2+h^{-1}_i\\
 \textit{s.t.}~~~& h_i=\|\bsphi_i\|^2_2+\mu-\p^{\top}_{t,i}\mathbf{r}_{t,i};\\
 & \r_{t,i}=\Q^{-1}_{\cS_t}\p_{t,i};\\
 & \p_{t,i}=[\bsphi_{\cS_t\{1\}},\dots,\bsphi_{\cS_t\{t\}}]^\top\bsphi_i
  \end{split}
\end{equation}
{where $h_{i}$ and $\p_{t,i}$ are computed from input $\bphi$ without first computing $\Q$ compared to equation \eqref{eq:greedy-mse-shift}.} However, this formulation still require matrix-vector multiplications for evaluating one candidate.

{Suppose that the optimal sample in step $t$ is $i^*$.
Then $\cS_{t}=\cS_{t-1}\cup\{i^*\}$.
For candidate node in unsampled set $i\in\cS^c_{t}$, we can write}
\begin{equation}\label{eq:pti}
  \p_{t,i}=[\bphi^\top_{\cS_{t-1},:} ~~\bsphi_{i^*}]^\top\bsphi_i=[\p^\top_{t-1,i} ~~\bsphi^\top_{i^*}\bsphi_i]^\top
\end{equation}

Thus, using computed result $\p_{t_1,i}$ in last greedy step as warm start, the computation of $\p_{t,i}$ only requires once vector-vector multiplication  $\bsphi^\top_{i^*}\bsphi_i$.
Further, for $i\in\cS^c_t$,
\begin{equation}
\label{eq:rti}
\begin{split}
 & \r_{t,i}=\Q^{-1}_{\cS_t}\p_{t,i}=\Q^{-1}_{\cS_{t-1}\cup \{i^*\}}\p_{t,i}\\
 & =\left[\begin{array}{cc}
              \Q^{-1}_{\cS_{t-1}}+h^{-1}_{i^*}\r_{t-1,i^*}\r^\top_{t-1,i^*}& -h^{-1}_{i^*} \r_{t-1,i^*}\\
             -h^{-1}_{i^*}\r^\top_{t-1,i^*} & h^{-1}_{i^*}
          \end{array}\right]
          \left[\begin{array}{c}
             \p_{t-1,i}  \\
             \bsphi^\top_{i^*}\bsphi_i
          \end{array}\right]\\
 &=\left[\begin{array}{c}
             \r_{t-1,i}+h^{-1}_{i^*}\r_{t-1,i^*}\r^\top_{t-1,i^*}\p_{t-1,i}-h^{-1}_{i^*} \bsphi^\top_{i^*}\bsphi_i\r_{t-1,i^*} \\
             -h^{-1}_{i^*}\r^\top_{t-1,i^*}\p_{t-1,i}+h^{-1}_{i^*}\bsphi^\top_{i^*}\bsphi_i
          \end{array}\right]\\
          &=\left[\begin{array}{c}
                    \r_{t-1,i}+\alpha\r_{t-1,i^*}-\beta\r_{t-1,i^*} \\
                    -\alpha+\beta
                  \end{array}\right]
\end{split}
\end{equation}
where
\begin{equation}
    \alpha=h^{-1}_{i^*}\p^\top_{t-1,i}\r_{t-1,i^*}
\end{equation}
\begin{equation}
    \beta=h^{-1}_{i^*} \bsphi^\top_{i^*}\bsphi_i
\end{equation}
and
\begin{equation}
\label{eq:hi}
   h_{i^*}=q_{i^*i^*}-\p^{\top}_{t-1,i^*}\r_{t-1,i^*}
\end{equation}

\textbf{Remark:} Based on equations \eqref{eq:rti} to \eqref{eq:hi}, for computing $\r_{t,i}$, we need  to reuse the computed warm starts in the last greedy step $h_{i^*}$, $\r_{t-1,i}$ and $\p_{t-1,i}$ and compute two new vector-vector multiplications, \textit{i.e.}, $\p^\top_{t-1,i}\r_{t-1,i^*}$ and  $\bsphi^\top_{i^*}\bsphi_i$.

Therefore, according to equations \eqref{eq:pti} and \eqref{eq:rti}, the computations of $\p_{t,i}$ and $\r_{t-i}$ only require four  vector-vector multiplications for evaluating one candidate. We write the complete greedy procedure in Algorithm \ref{algo:FMBS}, given matrix $\bphi$ as input, where subscript $t$ is abbreviated for simplicity.
\begin{table}
      \caption{Computational complexity comparisons of different sensor placment methods}
    \centering
    \label{tab:complexity}
    \begin{scriptsize}
    \begin{tabular}{|c|c|c|c|}
    \hline
    Method &{Convex}\cite{joshi2008sensor} / {SparSenSe}\cite{sparsense}&{FrameSense}\cite{linear-inverse}&{MNEP}\cite{jiang2016sensor}\\\hline
   Complexity &  $\cO(i_cN^3)$ / $\cO(i_sN^3)$&$\cO(N^3)$&$\cO(NMK^3)$\\\hline
   Method & {MPME}\cite{jiang2016sensor}&{fastMSE}\cite{jiang2019group}&{FMBS}\\\hline
  Complexity &  $\cO(NMK^2)$&$\cO(NMK^2)$&$\cO(NM^2)$\\\hline
    \end{tabular}
    \end{scriptsize}
\end{table}

{\subsection{Complexity Analysis}\label{subsec:complexity}}
Given sampling budget $M$ and the number of unsampled candidates $\cO(N)$, the complexity of the proposed method is $\cO(NM^2)$, because the complexity of vector-vector multiplications in each greedy search is at most $\cO(M)$.
{We call this method \textit{Fast MSE-based Sampling} (FMBS).}
We compare the proposed method with the following popular methods:
convex relaxation-based ({Convex}) \cite{joshi2008sensor}, sparse-aware sensor selection (SparSenSe) \cite{sparsense}, FrameSense  \cite{linear-inverse}, minimum nonzero eigenvalue pursuit ({MNEP}) \cite{jiang2016sensor}, maximal projection on minimum eigenspace ({MPME}) \cite{jiang2016sensor} and fast MSE pursuit-based ({fastMSE}) sampling \cite{jiang2019group}.
The computational complexities of those methods and the proposed FMBS are illustrated in Table.\;\ref{tab:complexity}, where some results are borrowed from \cite{jiang2016sensor} and \cite{jiang2019group}.
{The parameter $i_c$ is the iteration number of the interior-point method used in paper \cite{joshi2008sensor}, which is typically tens.}
Similarly, $i_s$ is the iteration number in the SparSenSe method.
Table.\;\ref{tab:complexity} shows that our proposed algorithm has the lowest theoretical complexity.
The empirical execution time and performance comparisons of these methods will be discussed in Section \ref{sec:experiment}. } 

\vspace{0.2cm}
\section{Experimentation}
\label{sec:experiment}
\begin{figure}
    \centering
            \subfigure{
    \includegraphics[width=115pt,height=115pt]{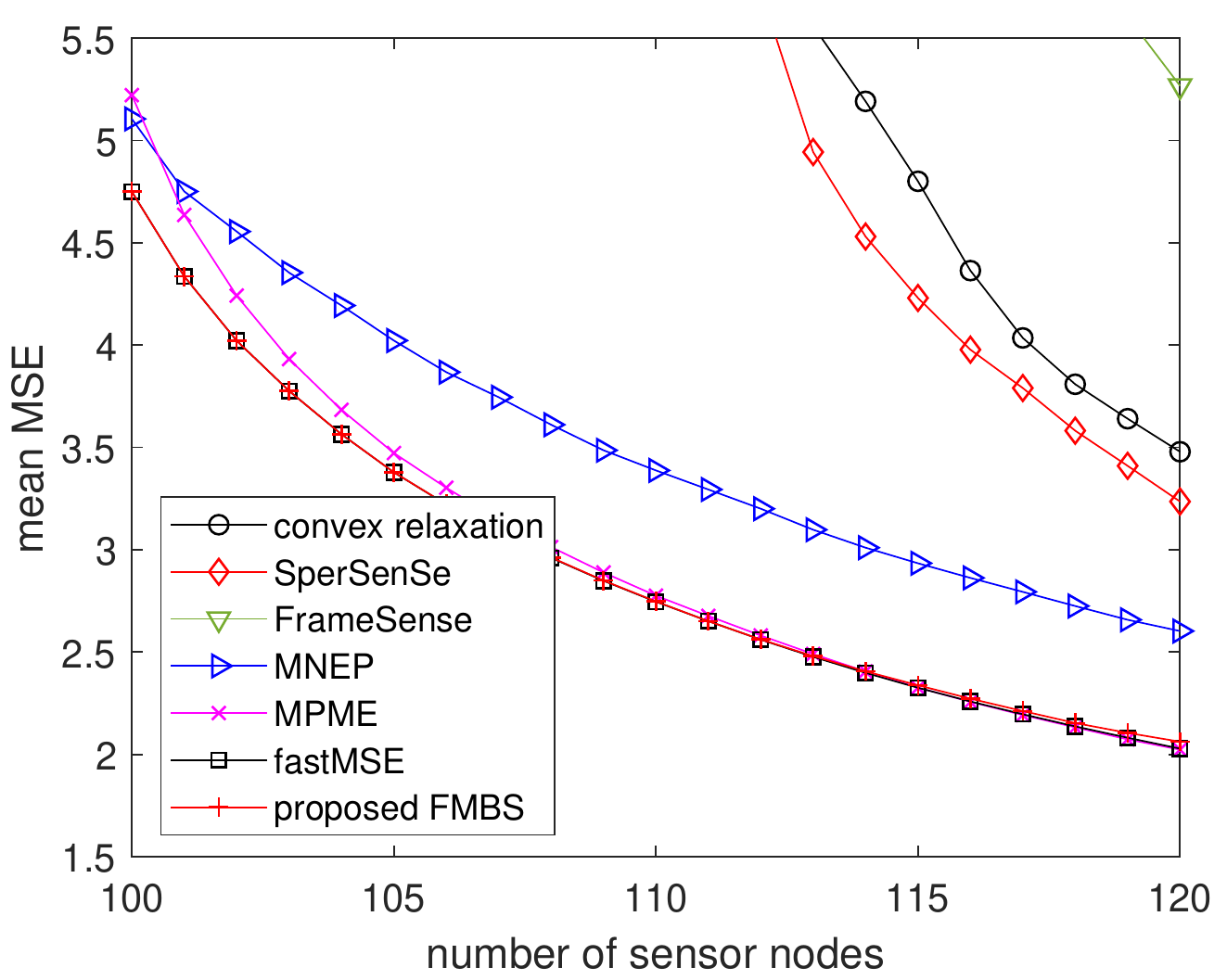}}
        \subfigure{
    \includegraphics[width=115pt,height=113pt]{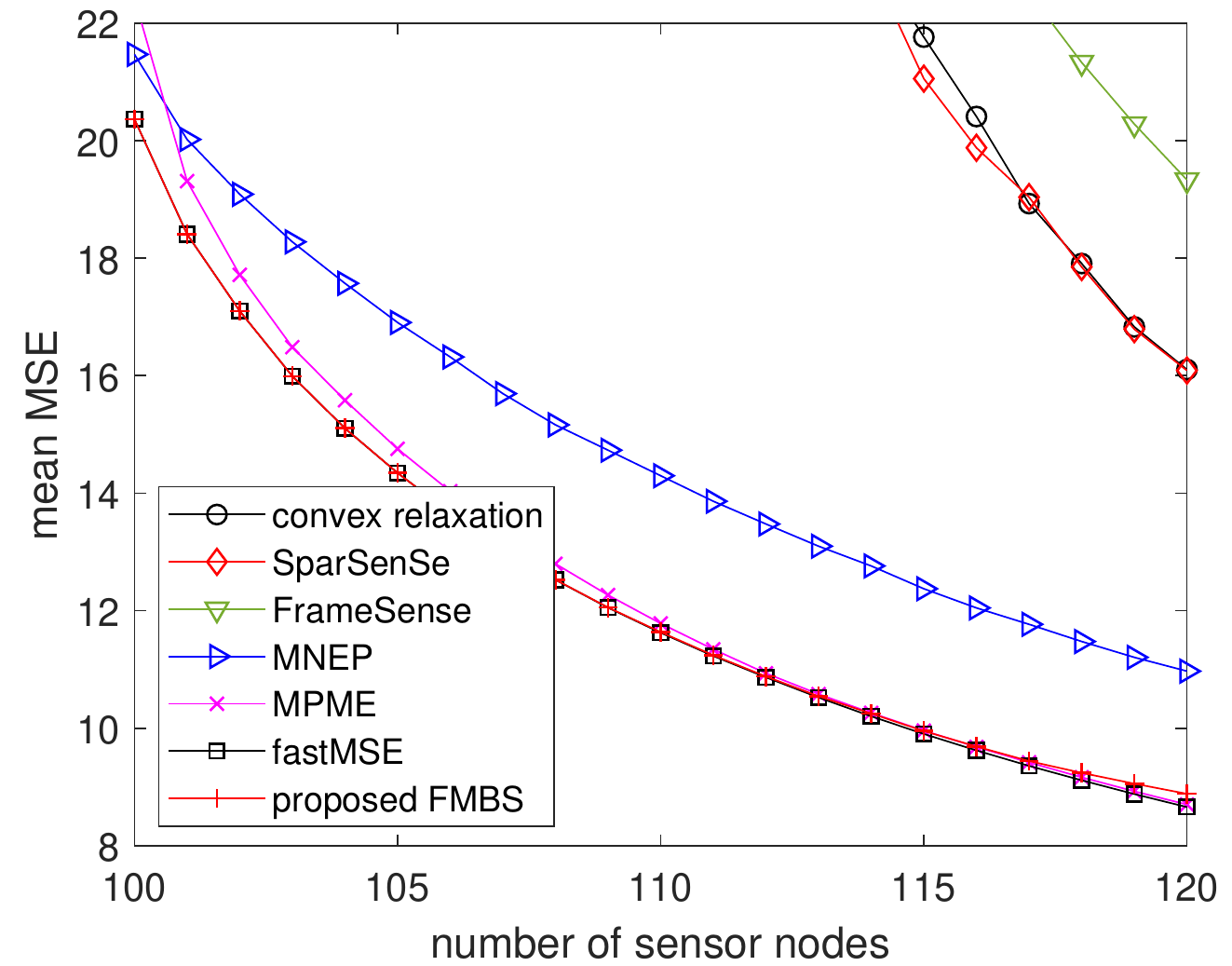}}
    \caption{Mean MSE of different sensor placement algorithms in terms of sensor budget, where the left measurement matrix $\bphi$ is generated from model 1 and the right one is generated from model 2.}
    \label{fig:mse}
\end{figure}
In this section, we will present extensive experimental results to evaluate the performance and efficacy of the proposed FMBS method. 
The measurement matrix $\bphi\in\bbR^{N\times K}$ in simulations can be generated from the following models \cite{jiang2016sensor}:

Model 1: $\bphi\in\bbR^{1000\times 100}$ is a Gaussian random matrix with
independent components $\phi_{ij}\sim\mathcal{N}(0,1)$, and the variance of
the noise $\sigma^2=1$;

Model 2:
$\bphi\in\bbR^{1000\times 100}$ is a Bernoulli random matrix with
independent components $\phi_{ij}\sim\mathcal{B}(1,1.5)$ with $\mathcal{B}$ representing the Binomial distribution, and the variance of the noise
$\sigma^2=1$;

\begin{table}[tp]
      \caption{Execution time for different sensor placement methods in terms of sensor budget $M$, when $N=1000$ and $K=100$.}
    \centering
    \label{tab:time}
    \begin{tabular}{cccccc}
    \hline
    ~&100&105&110&115&120\\\hline
    {Convex}\cite{joshi2008sensor}&33.45&32.04&31.70&31.38&31.65\\
    {FrameSense}\cite{linear-inverse}&0.217&0.214&0.206&0.211&0.207\\
    {MNEP}\cite{jiang2016sensor}&45.17&53.11&61.20&69.08&76.88\\
    {MPME}\cite{jiang2016sensor}&0.072&0.078&0.084&0.090&0.098\\
    {fastMSE}\cite{jiang2019group}&0.058&0.062&0.063&\textbf{0.067}&\textbf{0.069}\\
        {FMBS}&\textbf{0.056}&\textbf{0.059}&\textbf{0.062}&{0.068}&{0.072}\\\hline
    \end{tabular}
\end{table}

In our paper, we compute the averaged MSE value for each sampling algorithm via \cite{jiang2016sensor}:
\begin{equation}
  \overline{\textrm{MSE}}=\frac{1}{L}\sum_{i=1}^{L}\Tr\left[(\C_i\bphi_i)^\top(\C_i\bphi_i)\right]^{-1}
\end{equation}
where $L$ is the number of Monte-Carlo simulations and $\C_i$ is the sampling matrix obtained by different sampling methods at the $i$-th Monte-Carlo simulation.

We simulated the algorithms listed in Table.\;\ref{tab:complexity} for performance and execution time comparisons. The parameter $\mu$ was set to be $10^{-4}$ for the proposed method and fastMSE method. Fig.\;\ref{fig:mse} demonstrates the mean MSE performance of all simulated algorithms after 10 trails.
Fig.\;\ref{fig:mse} validates that the proposed FMBS achieved the same performance as fastMSE method, which was  better than all other competing popular  methods, especially when the sampling budget was small.
As stated after equation \eqref{eq:mse-shift}, fastMSE sampling  optimized the problems \eqref{eq:mse-shift} directly by greedy search, which has the same greedy solution as solving problem \eqref{eq:sampling-submatrix}. 

To evaluate experimental complexities of different methods, we recorded the execution time of different sampling methods under model 1, and illustrated the results in Table.\;\ref{tab:time}, where the best results are marked in bold. 
For this simulation, all experiments were performed on a laptop with Intel Core i7-8750H and 16GB of RAM on Windows 10 for counting time.
Table.\;\ref{tab:time} shows that the proposed {FMBS} cost the least sampling time when sample size is less than 110, and is comparable to fastMSE when $M=115,120$.

Further, we recorded the sampling time of MNEP, MPME, fastMSE and FMBS sampling methods on measurement matrix generated from model 1 with size $N=1000$, but with varied $K$. Convex and FrameSense sampling were not simulated since their performanc were not competitive. 
We set $M=K$ to fulfill the full column-rank requirement. The averaged execution time is  presented in Table.\;\ref{tab:timeK}, which indicates the proposed {FMBS} was the fastest one among simulated schemes. 

At last, for placing sensors on large-scale field, we performed relatively fast methods MPME and fastMSE as comparisons, where the size of target signal $N$ was vering from $2\times 10^3$ to $10^4$, and sampling budget $M$ is set to be $10\%N$. 
The dimension of parameter vector was $K=M$. 
Sampling time results of simulated methods were presented in Table.\;\ref{tab:timeN}, where the best results were marked in bold. 
It can be observed from Table.\;\ref{tab:timeN} that our proposed sensor placement method had the least sampling time in large-scale problems.

\vspace{-0.1cm}
\section{Conclusion}
\begin{table}[tp]
      \caption{Execution time for different sensor placement methods sensor networks with $N=1000$ and $K=M$.}
    \centering
    \label{tab:timeK}
    \begin{tabular}{cccccc}
    \hline
    ~&100&200&300&400&500\\\hline
    {MNEP}\cite{jiang2016sensor}&44.72&386.7&1246&2864&59036\\
    {MPME}\cite{jiang2016sensor}&0.074&0.609&1.964&4.448&8.468\\
    {fastMSE}\cite{jiang2019group}&0.057&0.541&1.503&2.942&5.260\\
    {FMBS}&\textbf{0.050}&\textbf{0.346}&\textbf{0.825}&\textbf{1.500}&\textbf{2.427}\\\hline
    \end{tabular}
\end{table}

\begin{table}[tp]
      \caption{Execution time for different sensor placement methods with varying $N$ and $M=K=10\%N$.}
    \centering
    \label{tab:timeN}
    \begin{tabular}{cccccc}
    \hline
    ~&2e3&4e3&6e3&8e3&1e4\\\hline
    {MPME}\cite{jiang2016sensor}&\textbf{0.916}&\textbf{8.91}&34.5&93.9&187\\
    {fastMSE}\cite{jiang2019group}&1.09&11.7&42.9&113&217\\
    {FMBS}&{1.02}&{9.45}&\textbf{33.3}&\textbf{82.0}&\textbf{156}\\\hline
    \end{tabular}
\end{table}
{In this paper, we proposed a fast sensor placement method for large-scale linear inverse problems. 
Specifically, assuming the field signal $\f$ is modelled by a linear equation $\f=\bphi{\g}$, it can be estimated from partial noisy samples via an unbiased least-square (LS) method, whose expected mean square error (MSE) is a function of measurement matrix $\bphi$ and sample set.
We formulated an approximate MSE problem, and then proved it is equivalent to a problem of a sample-dependent principle submatrix of $\bphi\bphi^\top$.
We proposed a fast greedy algorithm without matrix inverse computation. 
To further reduce complexity, we reused results in the previous greedy step for warm start to evaluate each candidate via few vector-vector multiplications.
Extensive experiments validated the performance and speed superiority of our proposed sensor placement method.}


\bibliographystyle{IEEEtran}
\bibliography{ref}
\end{document}